\documentclass[journal]{IEEEtran}
\usepackage{epsfig,latexsym,amsmath,amssymb,setspace,cite,color,graphicx,algorithm,algorithmic,cases}
\usepackage[caption=false,font=footnotesize]{subfig}
\usepackage[percent]{overpic}
\usepackage{stackrel}

\usepackage{paralist}
\usepackage{amsmath}
\usepackage{amsfonts}
\usepackage{amssymb}
\usepackage{dsfont}
\usepackage{bm}
\usepackage{amsthm}
\usepackage{newlfont}
\usepackage{float}
\usepackage{hyperref}
\usepackage{algorithm}
\usepackage{algorithmic}
\usepackage{enumerate}
\usepackage{chngcntr}
\usepackage{mathtools}
\usepackage{breqn}
\usepackage{stmaryrd}
\usepackage{cite}
\usepackage{amsmath}
\usepackage{amsfonts}
\usepackage{amssymb}
\usepackage{dsfont}
\usepackage{bm}
\usepackage{newlfont}
\usepackage{float}
\usepackage{cite}
\usepackage{algorithm}
\usepackage{algorithmic}
\usepackage{enumerate}
\usepackage{chngcntr}
\usepackage{mathtools}
\usepackage{breqn}
\usepackage{stmaryrd}

\usepackage{acronym}
 \hypersetup{
    colorlinks=true, %
    linkcolor= blue, %
    citecolor=blue %
}
\usepackage{amsmath} 
  
\newtheorem{thm}{Theorem} %
\newtheorem{lem}{Lemma}
\newtheorem{prop}{Proposition}

\newtheorem{defn}{Definition}
\newtheorem{rem}{Remark}

\title{Universal Covertness for\\ Discrete Memoryless Sources}

\author{R\'{e}mi A. Chou, Matthieu R. Bloch, and Aylin Yener%

 \thanks{R\'{e}mi A. Chou is with the Department of Electrical Engineering and Computer Science, Wichita State University, Wichita, KS 67260. Matthieu R. Bloch is with the School~of~Electrical~and~Computer~Engineering,~Georgia~Institute~of~Technology, Atlanta,~GA~30332--0250. Aylin Yener is with the Department of Electrical and Computer Engineering, The Ohio State University, Columbus, OH 43210. This work was supported in part by NSF under grants CIF-1319338, CNS-1314719,  CIF-1527074. Part of this work has been presented at the 54th Annual Allerton Conference on Communication, Control, and Computing~\cite{chou2016universal}. E-mails: remi.chou@wichita.edu, matthieu.bloch@ece.gatech.edu, yener@ece.osu.edu.
}

}

\DeclareMathOperator*{\argmax}{arg\,max}
\DeclareMathOperator*{\argmin}{arg\,min}
\DeclareMathOperator*{\plim}{p-lim}
\DeclareMathOperator*{\plimsup}{p-lim~sup}

\begin{document}
\sloppy
\allowdisplaybreaks[1]

\newenvironment{example}[1][Example]{\begin{trivlist}
\item[\hskip \labelsep {\bfseries #1}]}{\end{trivlist}}

\acrodef{DMS}[DMS]{Discrete Memoryless Source}
\renewcommand{\leq}{\leqslant}
\renewcommand{\geq}{\geqslant}

\maketitle

\begin{abstract}
Consider a sequence $X^n$ of length $n$ emitted by a Discrete
Memoryless Source (DMS) with unknown distribution $p_X$. The objective
is to construct a lossless source code that maps $X^n$ to a sequence
$\widehat{Y}^m$ of length $m$ that is indistinguishable, in terms of
Kullback-Leibler divergence, from a sequence emitted by another DMS with
known distribution $p_Y$. The main result is the existence of a coding
scheme that performs this task with an optimal
ratio $m/n$ equal to  $H(X)/H(Y)$, the ratio of the Shannon entropies of the two distributions, as $n$ goes to infinity. The coding scheme
overcomes the challenges created by the lack of knowledge about $p_X$
by a type-based universal lossless source coding scheme that produces as output an almost uniformly distributed sequence, followed by
another type-based coding scheme that jointly
performs source resolvability and universal lossless source
coding. The result recovers and extends  previous results that
either assume $p_X$ or $p_Y$ uniform, or $p_X$ known. The price paid
for these generalizations is the use of common randomness with
vanishing rate, whose length scales as the logarithm of
$n$. %
By allowing common randomness larger than the logarithm of $n$ but still
negligible compared to $n$, a constructive low-complexity encoding and
decoding counterpart to the main result is also provided for binary
sources by means of polar codes.%

\end{abstract}
 
 \section{Introduction}
\label{sec:introduction}

We consider the problem illustrated in Figure~\ref{fig1}, in which $n$ realizations of a \ac{DMS} $(\mathcal{X},p_X)$, with finite alphabet $\mathcal{X}$ and \emph{unknown} distribution $p_X$, are to be encoded into a vector $\widehat{Y}^m$ of length $m$. While $m$ should be as small as possible, the vector $\widehat{Y}^m$ should not only allow asymptotic lossless reconstruction of $X^n$ but also be asymptotically indistinguishable, in terms of Kullback-Leibler (KL) divergence, from a sequence $Y^m$ emitted by a \ac{DMS} $(\mathcal{Y},p_Y)$, with finite alphabet $\mathcal{Y}$ and known distribution $p_Y$. We refer to this problem as \emph{universal covertness} for \acp{DMS}, since an adversary observing $\widehat{Y}^m$ would then be unable to distinguish $\widehat{Y}^m$ from the output of the \ac{DMS} $(\mathcal{Y},p_Y)$. The formal relation between the~closeness of the distribution of $\widehat{Y}^m$ to the distribution of $Y^m$ and the probability of detection by the adversary follows from standard results on hypothesis testing, e.g., \cite{maurer2000authentication,blahut1987principles}. 

\begin{figure}
 \centering
  \includegraphics[width=8.2cm]{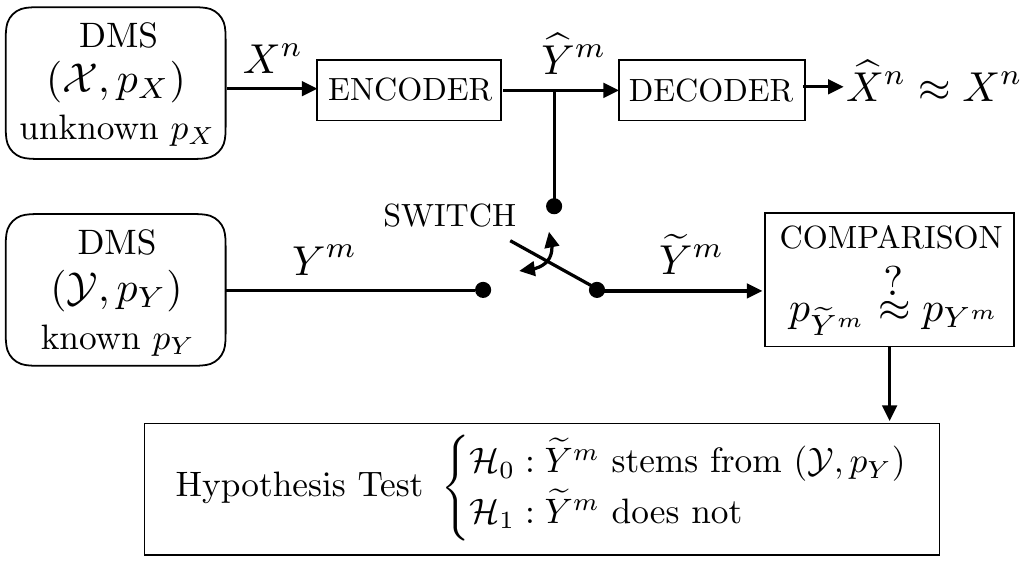}
  \caption{Illustration of universal covertness for DMSs.}
  \label{fig1}
 \end{figure}

Universal covertness generalizes and unifies several notions of random number generation and source coding found in the literature. For instance, 
\begin{inparaenum}[1)]
\item uniform lossless source coding~\cite{Chou13} corresponds to known $p_X$ and uniform $p_Y$;
\item random number conversion and source resolvability~\cite{HanBook,kumagai2013new} correspond to known $p_X$ and no reconstruction constraint;
\item universal source coding~\cite{bookCsizar} is obtained with known source entropy $H(X)$ and without distribution approximation constraint;
\item universal random number generation~\cite{oohama07} is obtained with known source entropy $H(X)$, uniform $p_Y$ and without  reconstruction constraint.
\end{inparaenum}
Universal covertness may also be viewed as a universal and noiseless counterpart of covert communication over noisy channels~\cite{bash2013limits,bloch2015covert,wang2015fundamental}. Most importantly, universal covertness relates to information-theoretic studies of information hiding and steganography~\cite{wang2008perfectly,cachin2004information}, yet with several notable differences that we now highlight.
\begin{itemize}
\item The problem in \cite{wang2008perfectly} consists in embedding a uniformly distributed message into a covertext without changing the covertext distribution, under a distortion reconstruction constraint. Universal covertness omits the distortion reconstruction constraint but relaxes the assumption of uniformly distributed message; this is motivated by the fact that message distributions encountered in practice are seldom uniform, and even optimally compressed data is only uniform in a weak sense~\cite{Han05,Hayashi08}. We point out that the perfect undetectability requirement enforced in~\cite{wang2008perfectly} is stronger than our asymptotic indistinguishability but largely relies on the presence of a long shared secret key.
\item The setting in~\cite[Section 4]{cachin2004information} is similar to universal covertness but does not address the problem of obtaining an optimal compression rate $m/n$ and indistinguishability is only measured in terms of \emph{normalized} KL-divergence. The extension in~\cite[Section 5]{cachin2004information} assumes that, unlike the adversary, the encoder only knows the entropy $H(Y)$ of the covertext and explicitly addresses the problem of estimating $p_Y$ from $n$ samples of the \ac{DMS} $(\mathcal{Y},p_Y)$. We recognize that, in practice, the covertext distribution $p_Y$ should be estimated from a finite number of samples, which necessarily limits the precision of the estimation. We take the view that the samples are public and in sufficient number so that all parties obtain the same estimates within an interval of confidence whose length is negligible compared to the uncertainty when estimating $p_Y$ from $m$ symbols of a \ac{DMS}.  %
\item As in~\cite{cachin2004information,wang2008perfectly}, universal covertness relies on a \emph{seed}, i.e., common randomness shared by the encoder and the decoder only; however, we shall see that the seed length used in our coding scheme is $\Theta(\log n)$, which is negligible compared to~$n$. This contrasts with seed lengths $\Theta(n \log n)$ in~\cite{wang2008perfectly} and $\Theta(n)$ in~\cite{cachin2004information}, although it is fair to mention that these larger key sizes enable perfect undetectability or perfect secrecy, which we do not require.
\end{itemize}
Beyond the generalizations offered by universal covertness described above, we note that the special case of uniform lossless source coding for \acp{DMS} with \emph{unknown distributions}, i.e., the case when $p_Y$ is uniform, is of particular interest to the design of secure communication schemes in settings where the uniformity of messages transmitted over a network is often a key assumption~\cite{Wyner75,Csiszar78}.

The idea of our proposed coding scheme is to approach the problem of universal covertness in two steps. In Step 1, universal uniform lossless source coding is performed through a type-based source coding scheme that makes the encoder output almost uniform. In Step 2, source resolvability with the additional constraint that the input be reconstructed from the output is performed with the result of Step 1 as input, so that the output of Step 2 approximates a given target distribution and allows recovery of the input of Step 1.

 We formally describe the problem in Section \ref{sec:probstat}. We study the special case of uniform  lossless source coding for \acp{DMS} with unknown distribution in Section \ref{sec:model}. Building upon the results of Section~\ref{sec:model}, we present our main result for universal covertness in Section \ref{sec:mres}. By allowing a larger amount of common randomness, whose rate still vanishes with the blocklength, we provide a constructive and low-complexity encoding and decoding scheme for universal covertness in Section~\ref{sec:con}. Finally, we  provide concluding remarks in Section~\ref{sec:concl}.
\section{Problem statement} \label{sec:probstat}
\subsection{Notation and basic inequalities} \label{sec:not}
 For $a,b \in \mathbb{R}_+$, we define $\llbracket a, b\rrbracket \triangleq [\lfloor a \rfloor, \lceil b \rceil] \cap \mathbb{N}$. For two functions $f$, $g$ from $\mathbb{N}$ to $\mathbb{R}_+$, we use the standard notation $f(n)=o(g(n))$ if $\lim_{n\to \infty} f(n)/g(n)=0$, $f(n)=O(g(n))$ if $\limsup_{n\to \infty} f(n)/g(n) < \infty$, and $f(n)=\Theta(g(n))$ if $\limsup_{n\to \infty} f(n)/g(n) < \infty$ and $\liminf_{n\to \infty} f(n)/g(n) >0$.
 For two distributions $p$ and $q$ defined over a finite alphabet $\mathcal{X}$, we define the variational distance $\mathbb{V}(p, q)\triangleq \sum_{x\in \mathcal{X}} |p(x) - q(x)|$, and denote the KL-divergence between $p$ and $q$ by $\mathbb{D}(p \lVert q)$ with the convention $\mathbb{D}(p \lVert q) = + \infty$ if there exists $x\in\mathcal{X}$ such that $q(x) =0$ and $p(x)>0$. Unless otherwise specified, capital letters denote random variables, whereas lowercase letters represent realizations of associated random variables, e.g., $x$ is a realization of the random variable $X$. We denote the indicator function by $\mathds{1}\{ \omega \}$, which is equal to $1$ if the predicate $\omega$ is true and $0$ otherwise. For any $x\in \mathbb{R}$, we define $[x]^+ \triangleq \max(0,x)$. For a sequence of random variables $(Z_n)_{n \in \mathbb{N}}$ that converges in probability towards a constant $C$, i.e., for any $\epsilon>0$, $\lim_{n\to\infty}\mathbb{P}\big(|Z_n-C| > \epsilon) = 0$, we use the notation $\displaystyle\plim_{n \to \infty} Z_n = C$.   We will also  use the following inequalities for KL-divergence, Eq.~\eqref{eqKL1} is from \cite{sason2015bounds}, Eq.~\eqref{eqKL3} is from \cite{chou2016empirical}, and Eq. \eqref{eqKL2} can easily be derived from the definition of the KL-divergence and Pinsker's inequality.
\begin{lem} \label{corent}
Let $p$, $q$, $r$ be distributions over the finite alphabet $\mathcal{X}$. %
Let $H(p)$ and $H(q)$ denote the Shannon entropy associated with $p$ and $q$, respectively. Let $\mu_q \triangleq \displaystyle\min_{x \in \mathcal{X}} q(x)$.
 We have
\begin{equation}
\mathbb{D}\left(p \Vert q \right)  \leq  \log \left(\mu_q^{-1} \right)\mathbb{V}\left(p , q \right),  \label{eqKL1} 
\end{equation}
\begin{multline}
  H(q) - H(p) \\ 
  \leq \mathbb{D}(p\lVert q) + \log \left(\mu_q^{-1} \right) \sqrt{2 \ln2} \sqrt{ \min ( \mathbb{D}(p\lVert q), \mathbb{D}(q\lVert p)) }, \label{eqKL2}
 \end{multline}
\begin{multline}
\mathbb{D}\left(p \Vert q \right)  \leq  \log \left(\mu_q^{-1} \right) \sqrt{2 \ln 2} \left[ \sqrt{ \min(\mathbb{D}\left(p \Vert r\right), \mathbb{D}\left(r \Vert p \right))} \right.  \\
  \left.  \phantom{-----}+\sqrt{ \min(\mathbb{D}\left(q \Vert r\right), \mathbb{D}\left(r \Vert q \right))} \right]. \label{eqKL3}
\end{multline}
\end{lem}

\subsection{Model for universal covertness} \label{sec:covm}
Consider a discrete memoryless source $( \mathcal{X},p_X )$.  Let $n \in \mathbb{N}$, $d_n \in \mathbb{N}$, and let $U_{d_n}$ be a uniform random variable over $\mathcal{U}_{d_n} \triangleq \llbracket 1, 2^{d_n}\rrbracket$, independent of $X^n$. In the following we refer to $U_{d_n}$ as the \emph{seed} and $d_n$ as its length. As illustrated in Figure~\ref{figmodel}, our objective is to design a source code to compress and reconstruct the  source $(\mathcal{X},p_{X})$, whose distribution is unknown, with the assistance of a seed $U_{d_n}$ and such that the encoder output approximates a known target distribution $p_Y$ with respect to the KL-divergence. 

\begin{figure}
\centering
  \includegraphics[width=7.7cm]{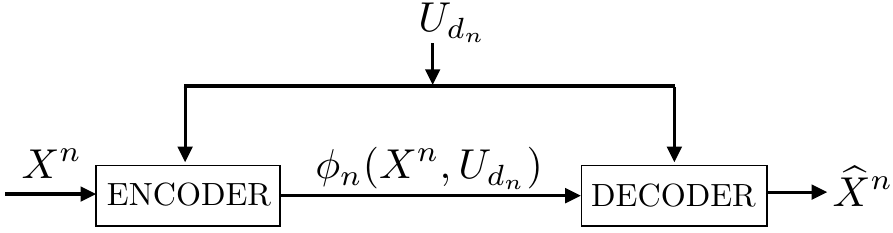}
  \caption{Universal covertness assisted with a seed (common randomness).}
  \label{figmodel}
\end{figure}
\begin{defn} \label{def_ucc}
An $(n,2^{d_n})$ variable-length universal covert source code for a DMS $(\mathcal{X},p_X)$ with respect to the DMS $(\mathcal{Y},p_Y)$ consists of 
\begin{itemize}
	\item A seed $U_{d_n}$ (with length $d_n$) uniformly chosen at random in the set $\mathcal{U}_{d_n} \triangleq \llbracket 1, 2^{d_n} \rrbracket$ and independent of all other random variables;
	\item An encoding function $\phi_n : \mathcal{X}^n \times \mathcal{U}_{d_n} \to \mathcal{Y}^m$ that takes as input the seed $U_{d_n}$ and the sequence $X^n$ emitted by the DMS $(\mathcal{X},p_X)$;
	\item A decoding function $\psi_n : \mathcal{Y}^m \times \mathcal{U}_{d_n} \to \mathcal{X}^n$. 
\end{itemize}
\end{defn}

\begin{rem}
We assume that $p_X$ is unknown; hence, $\phi_n$ and $\psi_n$ do not depend on prior knowledge about $p_X$  but are allowed to depend on the specific sequence of realizations of the DMS $(\mathcal{X},p_X)$, i.e.,  $(\phi_n,\psi_n)$ describes a variable-length code. Hence, 
 $m$ is a random variable that is a function of $X^n$  and is written as $m(X^n)$ to emphasize this point.
\end{rem}
The performance of a universal covert source code is measured in terms of 
\begin{enumerate}[(i)]
	\item The average probability of error $$\mathbb{P}[  X^n \neq \psi_n (\phi_n(X^n,U_{d_n}),U_{d_n})];$$
	\item Covertness, i.e., the closeness of the encoder output to a target distribution   $p^{\otimes m(X^n)}_{Y} \triangleq \prod_{i =1}^{m(X^n)} p_{Y}$, $$\mathbb{D} \left(p_{\phi_n(X^n,U_{d_n})} \lVert  p^{\otimes m(X^n)}_{Y} \right);$$ 
	\item Its output length to input length ratio ${m(X^n)}/n$, which should be minimized;
	\item The seed length $d_n$, which should be negligible compared to $n$.
\end{enumerate}

\begin{defn} \label{defcov}
Consider universal covertness for a DMS $(\mathcal{X},p_X)$ with respect to the DMS $(\mathcal{Y},p_Y)$. A rate $R$ is achievable if there exists a sequence of $(n,2^{d_n})$ variable-length universal covert source codes such that
\begin{align*}
\plimsup_{n \to \infty } \frac{m(X^n)}{n} &\leq R,\\
 \lim_{n \to \infty } \frac{d_n}{n} &=0,\\
\lim_{n \to \infty } \mathbb{P}[  X^n \neq \psi_n (\phi_n(X^n,U_{d_n}),U_{d_n})] &= 0,\\
\lim_{n \to \infty } \mathbb{D} \left(p_{\phi_n(X^n,U_{d_n})} \lVert  p^{\otimes m(X^n)}_{Y} \right) &=0 . 
\end{align*}
\end{defn}
We are interested in determining the infimum of all such achievable rates. %
\begin{rem}
Usually, for variable-length settings,   asymptotic average rates are considered (see, e.g.,~\cite{vembu1995generating,yagi2017variable} in the context of random number generation), i.e., convergence in mean is considered for coding rates. In this paper, we consider convergence in probability for the rate $m(X^n)/n$ for convenience, which also implies convergence in mean since the ratio $m(X^n)/n$ will be bounded in our setting. Our results will also show that the length of the encoder output concentrates with high probability around its optimal value $H(X)/ H(Y)$ for large $n$.
\end{rem}

\begin{rem}
Note that in the covertness condition, the term $\mathbb{D} \left(p_{\phi_n(X^n,U_{d_n})} \lVert  p^{\otimes m(X^n)}_{Y} \right)$ is a random variable as the length $m(X^n)$ of the encoder output is itself a random variable.
\end{rem}
\section{Special case: Uniform lossless source coding for \acp{DMS} with unknown distribution}\label{sec:model}
In this section, we study the problem described in Section~\ref{sec:covm}, in which $p_Y$ is the uniform distribution over $\mathcal{Y}$. We refer to this special case as uniform lossless source coding for \acp{DMS} with \emph{unknown distributions}. We build upon the solution proposed for this special case to provide a solution for the general case, i.e., arbitrary $p_Y$, in Section~\ref{sec:mres}.

The results of this section generalize and complement an earlier result for \acp{DMS} with known distributions~\cite{Chou13,Chou14b,Chou15d} when \emph{fixed-length} source coding is considered. 

\subsection{Definition of uniform lossless source coding} \label{sec:m}
For sources with unknown distributions, the problem of uniform lossless source coding aims at \emph{jointly} performing  universal lossless source coding \cite{bookCsizar,oohama1994universal} and universal randomness extraction~\cite{oohama07}.  More formally, universal uniform source coding is defined as follows.

\begin{defn} \label{def_ucc}
An $(n,2^{d_n})$ variable-length universal uniform source code is an $(n,2^{d_n})$ variable-length universal covert code for a DMS $(\mathcal{X},p_X)$ with respect to the DMS $(\{0,1\},p_U)$, where $p_U$ is the uniform distribution over $\{ 0,1\}$. We define its rate as $m(X^n)/n$.
\end{defn}
Similar to a universal covert code, the performance of a uniform source code is measured in terms of
\begin{enumerate}[(i)]
	\item The average probability of error $$\mathbb{P}[  X^n \neq \psi_n (\phi_n(X^n,U_{d_n}),U_{d_n})];$$
	\item The uniformity of its output $$ \mathbb{D} \left(p_{\phi_n(X^n,U_{d_n})} \lVert p_{U_{M_n(X^n)}} \right),$$ where $p_{U_{M_n(X^n)}}$ is the uniform distribution over $\mathcal{M}_n(X^n) \triangleq \llbracket 1, M_n (X^n)  \rrbracket$ with $M_n (X^n) \triangleq 2^{m(X^n)}$;
	\item The rate, which should be close to $H(X)$;
	\item The seed length $d_n$, which should be negligible compared to $n$.
\end{enumerate}
\begin{defn} \label{defusc}
Consider universal uniform source coding of a DMS $(\mathcal{X},p_X)$. A rate $R$ is achievable if there exists a sequence of $(n,2^{d_n})$ variable-length universal uniform source codes such that
\begin{align*}
\plimsup_{n \to \infty } \frac{m(X^n)}{n} &\leq R,\\
 \lim_{n \to \infty } \frac{d_n}{n} &=0,\\
\lim_{n \to \infty } \mathbb{P}[  X^n \neq \psi_n (\phi_n(X^n,U_{d_n}),U_{d_n})] &= 0,\\
\lim_{n \to \infty } \mathbb{D} \left(p_{\phi_n(X^n,U_{d_n})} \lVert p_{U_{M_n(X^n)}} \right) &=0 . 
\end{align*}
\end{defn}

\subsection{Method of types} \label{sec:csd}
We here recall known facts about the method of types~\cite{bookCsizar}.
Let $n \in \mathbb{N}$. For any sequence $x^n \in \mathcal{X}^n$, the type of $x^n$ is its empirical distribution given by $\left( \tfrac{1}{n}  \sum_{i=1}^n \mathds{1} \{ x_i =x \} \right)_{x \in \mathcal{X}} .$ Let $\mathcal{P}_n(\mathcal{X})$ denote the set of all types over $\mathcal{X}$, and $T_{\bar{X}}^n$ denote the set of sequences $x^n$ with type $p_{\bar{X}} \in \mathcal{P}_n(\mathcal{X})$. We will use the following lemma extensively.
\begin{lem} \label{lem1}
We have~\cite{bookCsizar}\\1) $|\mathcal{P}_n(\mathcal{X})| \leq (n+1)^{|\mathcal{X}|}$;	\\
	2) $ (n+1)^{-|\mathcal{X}|} 2^{nH(\bar{X})}  \leq |T^n_{\bar{X}}| \leq 2^{nH(\bar{X})}$;\\ 3) For $x^n \in T^n_{\bar{X}}$, $p_{X^n}(x^n) = 2^{-n (H(\bar{X}) + \mathbb{D}(p_{\bar{X}}\lVert p_X))}$.
\end{lem}
 
\subsection{Coding scheme} \label{sec:cs}
\subsubsection{Preliminary definitions}
 Consider $\alpha,\beta$ such that $\alpha> \beta >1$.
 Define $\gamma_n \triangleq  \lceil |\mathcal{X}|\log (n+1) \rceil$ and $a_n \triangleq 2^{\lceil \alpha \log n \rceil} \in \llbracket n^{\alpha},2 n^{\alpha} \rrbracket$. The seed is written as a pair $u_{d_n} \triangleq (u_1,u_2) \in \mathcal{U}_{d_n} \triangleq  \llbracket 0, a_n -1 \rrbracket \times \llbracket 0, 2^{\gamma_n} -1 \rrbracket$ with length $$d_n \triangleq \gamma_n + \lceil \alpha \log n \rceil = \Theta(\log n).$$  For any type $p_{\bar{X}} \in \mathcal{P}_n(\mathcal{X})$, define $$v_n(p_{\bar{X}}) \triangleq u_1\mod b_n(p_{\bar{X}}), $$ where $$b_n(p_{\bar{X}}) \triangleq \left\lfloor \frac{2^{\lceil \log |T_{\bar{X}}^n| + \beta \log n \rceil}}{|T_{\bar{X}}^n|}  \right\rfloor  \in \llbracket n^{\beta},2 n^{\beta} \rrbracket.$$
Next, consider an arbitrary injective mapping $s$
  \begin{align*}
 	s: \mathcal{P}_n(\mathcal{X}) &\to  \llbracket 0, 2^{\gamma_n} -1\rrbracket\\
 	 p_{\bar{X}} &\mapsto s(p_{\bar{X}}),
 \end{align*}
 which uniquely indexes the types in $\mathcal{P}_n(\mathcal{X})$. Finally, for any type $p_{\bar{X}} \in \mathcal{P}_n(\mathcal{X})$ with index $s(p_{\bar{X}}) \in \llbracket 0, 2^{\gamma_n} -1 \rrbracket$, consider an arbitrary one-to-one mapping $k_{s(p_{\bar{X}})}$
 \begin{align*}
 	k_{s(p_{\bar{X}})}: T_{\bar{X}}^n &\to \llbracket 0, |T_{\bar{X}}^n| -1 \rrbracket\\
 	 x^n &\mapsto k_{s(p_{\bar{X}})}(x^n),
 \end{align*}
which uniquely indexes the sequences of type $p_{\bar{X}}$. We also define for any type $p_{\bar{X}} \in \mathcal{P}_n(\mathcal{X})$, $$c_n (p_{\bar{X}})\triangleq 2^{\lceil \log |T_{\bar{X}}^n| +  \beta \log n \rceil},$$ and $$c_n(X^n) \triangleq c_n(p_{\bar{X}}(X^n)),$$  where $p_{\bar{X}}(X^n)$  is the type of $X^n$.

 \subsubsection{Encoder and decoder}

 The encoder is the map
 \begin{align*}
\phi_n :  \mathcal{X}^n \times \mathcal{U}_{d_n} & \rightarrow \llbracket 0,c_n(X^n)-1 \rrbracket \times \llbracket 0, 2^{\gamma_n} -1\rrbracket\\
 (x^n,(u_1,u_2)) &\mapsto \left(\phi^{(1)}_n(x^n,u_1), \phi^{(2)}_n(x^n,u_2)\right), 
\end{align*}
where for any $p_{\bar{X}} \in \mathcal{P}_n(\mathcal{X})$, for any $x^n \in T_{\bar{X}}^n$,
 \begin{align}
\phi^{(1)}_n (x^n,u_1) &\triangleq  k_{s(p_{\bar{X}})}(x^n)+  v_n(p_{\bar{X}})|T_{\bar{X}}^n| , \label{eqdith}\\
\phi^{(2)}_n (x^n,u_2) &\triangleq  s(p_{\bar{X}}) \oplus u_2,
\end{align}
where $s(p_{\bar{X}}) \oplus u_2$ denote the modulo-2 addition between the binary representation of $s(p_{\bar{X}})$ and $u_2$. Note that the idea in \eqref{eqdith} is to introduce a dithering term to make $\phi^{(1)}_n (x^n,u_1)$ almost uniform.

The decoder is the map $$\psi_n: \llbracket 0,c_n(X^n)-1 \rrbracket \times \llbracket 0 , 2^{\gamma(n)} -1 \rrbracket \times \mathcal{U}_{d_n} \rightarrow \mathcal{X}^n $$ with
\begin{align*}
\psi_n: 
(i,j,(u_1,u_2)) &\mapsto k^{-1}_{u_2 \oplus j}\left(i-v_n(p_{\bar{X}})|T_{\bar{X}}^n| \right)
\end{align*}
where $v_n(p_{\bar{X}})|T_{\bar{X}}^n|$ can be computed from the knowledge of the type index $j$ and $u_1$.

\begin{rem}
Since a given sequence $x^n \in \mathcal{X}^n$ is uniformly distributed in its class type, one could think of the following simpler encoding scheme.
Let $p_{\bar{X}} \in \mathcal{P}_n(\mathcal{X})$ and $x^n \in T_{\bar{X}}^n$.
Choose \begin{align*}
 	k_{s(p_{\bar{X}})}: T_{\bar{X}}^n &\to \llbracket 0, 2^{\lceil \log |T_{\bar{X}}^n| \rceil} -1 \rrbracket\\
 	 x^n &\mapsto k_{s(p_{\bar{X}})}(x^n),
 \end{align*}
to uniquely index the sequences of type $p_{\bar{X}}$ and define
 \begin{align*}
\phi^{(1)}_n (x^n,u_1) &\triangleq  k_{s(p_{\bar{X}})}(x^n).
\end{align*}
However, such an encoding scheme would not allow the encoder output to be almost uniformly distributed. Indeed, since 
$ |T_{\bar{X}}^n|    \leq  2^{\lceil\log |T_{\bar{X}}^n| \rceil} <  2 |T_{\bar{X}}^n| ,$
between $0$ and $|T_{\bar{X}}^n|$ values will never be taken by $k_{s(p_{\bar{X}})}(x^n)$, and one can show that, for some types, the number of values that are never taken is too large to allow uniformity of the encoder~output.
 \end{rem}

\subsection{Results} \label{sec:resdc}
\begin{thm} \label{Cor2}
The coding scheme of Section \ref{sec:cs} provides a sequence of $(n,2^{d_n})$ variable-length universal uniform source codes, such that for any DMS with unknown distribution $p_X$, we have 
\begin{align}
\plim_{n \to \infty } \frac{m(X^n)}{n} &\leq H(X),\label{c4}\\
 d_n  &= \Theta (\log n), \label{c1}\\
 \mathbb{P}[  X^n \neq \psi_n (\phi_n(X^n,U_{d_n}),U_{d_n})] &= 0, \label{c2} \\
\lim_{n \to \infty } \mathbb{D} \left(p_{\phi_n(X^n,U_{d_n})} \lVert p_{U_{M_n(X^n)}} \right) &=0 .  \label{c3}
\end{align}
\end{thm}
\begin{proof}
See Section \ref{sec_proof}.	
\end{proof}

\begin{prop}
The asymptotic rate $\displaystyle\plim_{n \to \infty}  \frac{m(X^n)}{n}$ in Theorem~\ref{Cor2} is optimal.
\end{prop}

\begin{proof} 
See Appendix \ref{App_converse}.
\end{proof}
\subsection{Proof of Theorem \ref{Cor2}} \label{sec_proof}

	By construction, we clearly have \eqref{c1} and \eqref{c2}. We now prove \eqref{c3} through a series of lemmas.
We first prove almost uniformity of the variables  	$V_n(p_{\bar{X}})$ that appear in the coding scheme.
	
	\begin{lem} \label{lem_V}
For any type $p_{\bar{X}} \in \mathcal{P}_n(\mathcal{X})$, $V_n(p_{\bar{X}})$ is almost uniformly distributed over $\llbracket 0, b_n(p_{\bar{X}}) -1\rrbracket$ in the sense that 
\begin{align*}
\mathbb{V} \left( p_{V_n(p_{\bar{X}}) }, p_{\bar{V}_n(p_{\bar{X}}) }\right) = O\left(\frac{1}{n^{\alpha-\beta}}\right),
\end{align*}
where $p_{\bar{V}_n(p_{\bar{X}}) }$ is the uniform distribution over $\llbracket 0, b_n(p_{\bar{X}}) -1 \rrbracket$.	
\end{lem}

\begin{proof}
	See Appendix \ref{App_V}.
	\end{proof}

Using Lemma \ref{lem_V}, we next prove that conditioned on the type of the sequence to compress, the encoder output $\phi^{(1)}_n(X^n,U_{d_n})$ approximates a uniform distribution with respect to the variational distance.

\begin{lem} \label{lemcond}
	For any type $p_{\bar{X}} \in \mathcal{P}_n(\mathcal{X})$, we have
\begin{align*}
\mathbb{V} \left( p_{\phi^{(1)}_n(X^n,U_{d_n})|X^n \in T_{\bar{X}}^n } , p_{U(p_{\bar{X}})}\right)		= O\left(\frac{1}{n^{\min(\alpha - \beta, \beta)}}\right),
\end{align*}
where	$ p_{U(p_{\bar{X}})}$ is the uniform distribution over   
$ \llbracket 0,c_n(p_{\bar{X}})-1 \rrbracket$.%
\end{lem}
\begin{proof}
	See Appendix \ref{App_lemcond}.
\end{proof}

Using Lemma \ref{lemcond}, we deduce the following result unconditional on the type of the sequence to compress.

\begin{lem} \label{lem_avg}
Define $p_{U(X^n)}$ as the uniform distribution over $\llbracket 0,c_n(X^n)-1 \rrbracket$. 
\begin{align*}
\mathbb{V} \left( p_{\phi^{(1)}_n(X^n,U_{d_n})} , p_{U(X^n)}\right)		= O\left(\frac{1}{n^{\min(\alpha - \beta, \beta)}}\right).
\end{align*}
\end{lem}
\begin{proof}
See Appendix \ref{App_avg}.
\end{proof}

From Lemma \ref{lem_avg}, which quantifies the uniformity of the encoder output $\phi^{(1)}_n(X^n,U_{d_n})$ in terms of variational distance, we deduce the following result which now quantifies uniformity in terms of KL-divergence.
\begin{lem} \label{lem_div}
We have
\begin{align*}
\mathbb{D} \left( p_{\phi^{(1)}_n(X^n,U_{d_n})} \lVert p_{U(X^n)} \right)		= O\left(\frac{1}{n^{\min(\alpha - \beta, \beta)-1}}\right).
\end{align*}

\end{lem}
\begin{proof}
See Appendix \ref{App_div}.
\end{proof}

Finally, let $p_{U_{M_n(X^n)}}$ be the uniform distribution over $\llbracket 0,c_n(X^n)-1 \rrbracket \times \llbracket 0, 2^{\gamma_n} -1\rrbracket$. We have
\begin{align*}
	\mathbb{D} \left(p_{\phi_n(X^n,U_{d_n})} \lVert p_{U_{M_n(X^n)}} \right)
	& \stackrel{(a)} = \mathbb{D} \left( p_{\phi^{(1)}_n(X^n,U_{d_n})} \lVert p_{U(X^n)} \right) \\
	& \stackrel{(b)} = O\left(\frac{1}{n^{\alpha/2 -1}}\right),
\end{align*}
where $(a)$ holds by the chain rule for divergence, independence between $p_{\phi^{(1)}_n(X^n,U_{d_n})}$ and $p_{\phi^{(2)}_n(X^n,U_{d_n})}$, and perfect uniformity of $\phi^{(2)}_n(X^n,U_{d_n})$ by construction, $(b)$ holds by choosing $\beta = \alpha /2$ in Lemma \ref{lem_div}. We thus conclude that \eqref{c3} holds by choosing $\alpha>2$. 

We now prove \eqref{c4}.  
Let $T_{\bar{X}}^n(X^n)$ denote the type set to which $X^n$ belongs. Let $\widehat{H}(X^n)$ denote the plug-in estimate of $H(X)$ using $X^n$ \cite{basharin1959statistical}. The encoder output length is  
\begin{align*}
\log (c_n (X^n) \times 2^{\gamma_n}) 
& \stackrel{(a)}\leq \log |T_{\bar{X}}^n(X^n)| + \beta \log n +1 + \gamma_n \\
& \stackrel{(b)} \leq n\widehat{H}(X^n) + \beta \log n +1 + \gamma_n \\
& \stackrel{(c)} \leq n\widehat{H}(X^n) + O( \log n), 
\end{align*}
where $(a)$ holds by definition of $c_n (X^n)$, $(b)$ holds by Lemma~\ref{lem1} and because $\widehat{H}(X^n)=H(p_{\bar{X}}(X^n))$, $(c)$ holds by the definition of $\gamma_n$. Hence, we conclude that \eqref{c4} holds by~\cite{basharin1959statistical}. 

\section{Covertness for \acp{DMS} with unknown distribution}
 \label{sec:mres}
Our coding scheme for universal covertness uses two building blocks, which are two special cases of the model described in Section~\ref{sec:covm}: (i) Uniform source coding for DMS with unknown distribution, studied in Section \ref{sec:model}; (ii)~Source resolvability with the additional constraint that the input should be recoverable from the output, which corresponds to the case in which $p_X$ is known to be the uniform distribution.

\subsection{Results}
Building upon our construction in Theorem \ref{Cor2} we obtain the following results.
\begin{thm} \label{Cor9}
There exists a sequence of $(n,2^{d_n})$ variable-length universal covert source code for the DMS $(\mathcal{X},p_X)$ with respect to the DMS $(\mathcal{Y},p_Y)$, defined by the encoding/decoding functions $(\phi_n,\psi_n)$ with encoder output length $m(X^n)$, such that if one defines $\widehat{Y}^{m(X^n)} \triangleq \phi_n(X^n,U_{d_n})$,~then 
\begin{align}
 \plim_{n \to \infty}  \frac{m(X^n)}{n}  &\leq {H(X)}/{H(Y)}, \label{d1}\\
  d_n & = \Theta ( \log n), \label{d2} \\
\mathbb{P}[  X^n \neq \psi_n (\widehat{Y}^{m(X^n)},U_{d_n})] &= 0,  \label{d3}\\
 \lim_{n \to \infty}  \mathbb{D} \left(p_{\widehat{Y}^{m(X^n)}} \lVert  p^{\otimes m(X^n)}_{Y} \right) &=0. \label{d4}
\end{align}
\end{thm}
\begin{proof}
	See Section \ref{sec:proof2}.
\end{proof}

\begin{prop}
The asymptotic rate  $\displaystyle\plim_{n \to \infty}  \frac{m(X^n)}{n}$ in Theorem~\ref{Cor9} is optimal.
\end{prop}

\begin{proof}
	See Appendix \ref{App_converse}.
\end{proof}
\subsection{Proof of Theorem \ref{Cor9}} \label{sec:proof2}
 We first perform source resolvability with lossless reconstruction of the input from the output by means of ``random binning" \cite{HanBook,yassaee2014achievability,nafea2018new}. Note that standard resolvability results~\cite{HanBook} do not directly apply to our purposes as they do not support the recoverability constraint of the input from the output.

Let $m \in \mathbb{N}$ to be specified later and define $R_Y \triangleq H(Y) - \epsilon$, $\epsilon>0$, where $H(Y)$ is the entropy associated with the target distribution $p_Y$.
To each $y^m \in \mathcal{Y}^m$, we assign an index $B(y^m) \in \llbracket 1 , 2^{mR_Y} \rrbracket$ uniformly at random. The joint probability distribution between $Y^m$ and $B(Y^m)$ is given by,
\begin{align} \label{eq:rb}
 &\forall y^m \in \mathcal{Y}^m, \forall b \in \llbracket 1 , 2^{mR_Y} \rrbracket, \nonumber\\& \phantom{---} p_{Y^m B(Y^m)}(y^m,b) = p_{Y^m}(y^m) \mathds{1}\{ B(y^m) = b \}.
\end{align}
We then consider the random variable $\widehat{Y}^m$ that is distributed according to $\widehat{p}_{{Y}^m}$, where
\begin{align}\label{eq:rb2}
&\forall y^m \in \mathcal{Y}^m, \forall b \in \llbracket 1 , 2^{mR_Y} \rrbracket, \nonumber\\& \phantom{---} \widehat{p}_{{Y}^mB(Y^m)}(y^m,b) \triangleq p_{Y^m|B(Y^m)}(y^m|b) p_{U}(b),
\end{align}
and $p_{U}$ is the uniform distribution over $\llbracket 1 , 2^{mR_Y} \rrbracket$. We thus have
\begin{align}
\mathbb{E}_{B}\mathbb{V}(p_{Y^mB(Y^m)},\widehat{p}_{{Y}^mB(Y^m)}) \nonumber  \nonumber
& = \mathbb{E}_{B}\mathbb{V}(p_{B(Y^m)},\widehat{p}_{B(Y^m)})  \\ 
& = \mathbb{E}_{B}\mathbb{V}(p_{B(Y^m)},p_U) \\
& = o(m^{-r}), \label{eq:rb3}
\end{align}
where the last equality holds by \cite[Theorem 1]{yassaee2014achievability} for any $r>0$. 
Observe also that when $\widehat{y}^m$ is drawn according to $\widehat{p}_{Y^m|B(Y^m)=b}$ for some $b$, then $b$ can be perfectly recovered from $\widehat{y}^m$, since by \eqref{eq:rb}, \eqref{eq:rb2}, we have
\begin{align}\label{eq:rb4}
B( \widehat{y}^m ) = b.
\end{align} 
All in all, \eqref{eq:rb3} and \eqref{eq:rb4} mean that there exists a specific choice $B_0$ for the binning $B$ such that, if $b$ is a sequence of length $mR_Y$ distributed according to $p_U$ and $\widehat{y}^m$ is drawn according to $p_{Y^m|B_0(Y^m)=b}$, then $B_0(\widehat{y}^m)=b$ and $\mathbb{V}(\widehat{p}_{{Y}^m},p_{Y^m})=o(m^{-r})$ for any $r>0$. By the triangle inequality, the result stays true if $p_U$ is replaced by a distribution $\widetilde{p}_U$ that satisfies $\mathbb{V}(\widetilde{p}_U,p_U) =o(m^{-r})$ for any $r>0$. Note that the construction requires randomization at the encoder, however, the randomness needs not be known by the decoder.

We now combine source resolvability with lossless reconstruction of the input from the output and universal uniform source coding as follows.
Let $n \in \mathbb{N}$, and consider a variable-length uniform source code obtained from Theorem~\ref{Cor2} and described by the encoding/decoding pair $\left( \phi_n, \psi_n \right)$ where $\phi_n(X^n,U_{d_n}) = (\phi^{(1)}_n(X^n,U_{d_n}),\phi^{(2)}_n(X^n,U_{d_n}))$ as described in Section \ref{sec:cs} and define $M_1 \triangleq \phi^{(1)}_n(X^n,U_{d_n})$, $M_2 \triangleq \phi^{(2)}_n(X^n,U_{d_n})$. %
We then define the length of our universal covert source encoder output such that $
  \left\lceil m(X^n) R_Y \right\rceil = |M_1| + |M_2| +T$ for some $T \in \llbracket 0, R_Y  \rrbracket$, where $|\cdot|$ denotes the length of a sequence. 
We also define the sequence
$M \triangleq (M_1 \lVert C \lVert M_2),$
where $\lVert$ denotes the concatenation of sequences and $C$ is a sequence of $T$ bits uniformly distributed.

Finally, the encoder of our universal covert source code forms $\widehat{Y}^{m(X^n)}$ by source resolvability as previously described with the substitutions $b \leftarrow M$, $m \leftarrow m(X^n)$, so that the decoder of our universal covert source code determines from $\widehat{Y}^{m(X^n)}$, in this order, $M$, then $M_2$ (since the length of $M_2$ is known to be $\gamma_n$), then $|M_1|$ (since $M_2$ reveals the type of the compressed sequence $X^n$ given the seed), then $M_1$, and finally approximates $X^n$ using $\psi_n$ applied to $(M_1,M_2,U_{d_n})$. Hence, \eqref{d2} and \eqref{d3} hold. %
\begin{rem}
Note that the sequence $C$ does not carry information and is only used to pad the sequence $M$ such that $|M| = \left\lceil m(X^n) R_Y \right\rceil$.%
\end{rem}
Next, we have
\begin{align*}
	\frac{m(X^n)}{n} 
	& \stackrel{(a)} \leq \frac{|M_1|+|M_2|+T}{nR_Y} \\
	& \stackrel{(b)} \leq \frac{|M_1|+|M_2|}{n} \frac{1}{H(Y)-\epsilon} + \frac{1}{n},
\end{align*}
where $(a)$ holds by definition of $m(X^n)$, and $(b)$ holds by definition of $R_Y$. Hence, by Theorem~\ref{Cor2}, $\plim_{n \to \infty} m(X^n)/n \leq H(X)/(H(Y) - \epsilon)$. Finally, since the encoder output of the universal source code is almost uniform, as described in Theorem \ref{Cor2}, we have also obtained 
$ \lim_{n \to \infty}  \mathbb{V} \left(p_{\widehat{Y}^{m(X^n)}} , p^{\otimes m(X^n)}_{Y} \right) =o(m^{-r})$ for any $r>0$, which implies \eqref{d4} by Lemma \ref{corent}.

\section{A Constructive and Low-Complexity Coding Scheme} \label{sec:con}

Theorem \ref{Cor2} provides a coding scheme for universal uniform source coding but the implementation of the coding scheme is intractable since it relies on the method of types.   As for Theorem~\ref{Cor9}, it only provides an existence result (i.e., a non-constructive coding scheme) for universal source covertness. We present in this section a constructive and low-complexity counterpart to Theorem~\ref{Cor2} and Theorem~\ref{Cor9} for a binary source alphabet, i.e., $|\mathcal{X}|=2$. The seed length required in our coding scheme will be shown to be negligible compared to the length of the compressed sequence but will be larger than the one in Theorems~\ref{Cor2}, \ref{Cor9}.

In Definition \ref{def_ucc}, assume that the DMS $(\mathcal{X},p_X)$ is Bernoulli with parameter $p \neq 1/2$, unknown to the encoder and decoder, and that the DMS $(\mathcal{Y},p_Y)$ is such that $|\mathcal{Y}|$ is a prime number. Let $n \in \mathbb{N}^*$, $N \triangleq 2^n$, and consider a sequence $x^{LN}$ of $LN$, where $L \in \mathbb{N}^*$ will be specified later, independent realizations of $(\mathcal{X},p_X)$ that need to be compressed. Let $\lVert$ denote the concatenation of sequences, $\backslash$ denote set subtraction,  and $H_b$ denote the binary entropy. Also, define $G_n \triangleq  \left[ \begin{smallmatrix}
       1 & 0            \\[0.3em]
       1 & 1 
     \end{smallmatrix} \right]^{\otimes N} $, the polarization matrix defined in~\cite{Arikan10}, and define for any set $\mathcal{I} \subseteq \llbracket 1,N \rrbracket $, for any sequence $X^N$, the subsequence  $X^{N}[\mathcal{I}] \triangleq \left( X_{i}\right)_{i \in \mathcal{I}}$, where $X_{i}$, $i \in \mathcal{I}$, denotes the $i$-th component of~$X^{N}$.

\subsection{Coding scheme} \label{sec:decp}

\textbf{Encoding}: We proceed in three steps. Step 1 is the estimation of $p$. Step 2 corresponds to universal uniform source coding. It is performed with polar codes and generalizes both the coding scheme in \cite{Chou14prep}, which cannot account for uncertainty on the source distribution, and the coding scheme in~\cite{chou2020explicit}, which can only account for a compound setting. Step~3 corresponds to source resolvability with lossless reconstruction of the input from the output. It is also performed with polar codes with methods similar to those used in \cite{chou2014polar} but with the additional difficulty that the exact length of the input is unknown to the decoder.

\textbf{Step 1.} %
Let $t < 1/2$ and define 
$
q \triangleq \lceil N^{t} \rceil$, $\delta \triangleq N^{-t}.
$ We also define $a_i \triangleq i \delta, i \in \llbracket 0 , q-1 \rrbracket$, $a_q \triangleq 1$, $a_{-1} \triangleq a_0$, and $a_{q+1} \triangleq a_q$ such that 
 $\{ [a_i, a_{i+1}]\}_{i \in \llbracket 0, q-1\rrbracket}$ is a partition of $[0, 1]$.
 We estimate $p$ as 
 \begin{align*}
 	\hat{p} \triangleq \frac{1}{LN}\sum_{i=1}^{LN} \mathds{1}\{ x_i = 1 \}.
 \end{align*}
 There exists $i_0 \in \llbracket 1,q\rrbracket$ such that 
$\hat{p} \in [a_{i_0-1}, a_{i_0}]$. Next, we define 
\begin{align}
	{\underline{p}} &\triangleq \displaystyle\argmax_{a \in\{a_{i_0-2},a_{i_0+1}\}}  |a-1/2|, \label{eqp1} \\
	{\overline{p}} &\triangleq \begin{cases} \displaystyle\argmin_{a \in\{a_{i_0-2},a_{i_0+1}\}}  |a-1/2| & \text{if } \frac{1}{2} \notin [a_{i_0-2},a_{i_0+1}]\\
	\frac{1}{2} & \text{if } \frac{1}{2} \in [a_{i_0-2},a_{i_0+1}]
 \end{cases}.  \label{eqp}
\end{align}
Let $I_{0}$ be the binary representation of $i_0$ and form 
\begin{align}
	I_{N} \triangleq {I}_{0} \oplus K_0, \label{eqp2}
\end{align}
 where $K_0$ is a sequence of uniform bits with length $\lceil \log(q) \rceil = O(\log N)$ that is shared by the encoder and decoder.

\textbf{Step 2.}
 Let $X^N$ ($\underline{X}^N$, $\overline{X}^N$, respectively) be a sequence of $N$ independent Bernoulli random variables with parameter $p$ (${\underline{p}}$, ${\overline{p}}$, respectively). We perform universal uniform source coding on $X^N$ in this second step. Define ${U}^N \triangleq {X}^N G_n$, $\underline{U}^N \triangleq \underline{X}^N G_n$, $\overline{U}^N \triangleq \overline{X}^N G_n$, and for $\beta < 1/2$, $\delta_N \triangleq 2^{-N^{\beta}}$, define the sets
 \begin{align*}
\mathcal{H}_{\underline{X}} &\triangleq \{ i \in \llbracket 1, N \rrbracket : H(\underline{U}_i|\underline{U}^{i-1}) \geq \delta_N \}, \\
\mathcal{V}_{\underline{X}} &\triangleq \{ i \in \llbracket 1, N \rrbracket : H(\underline{U}_i|\underline{U}^{i-1}) \geq 1- \delta_N \},\\
\mathcal{H}_{\overline{X}} &\triangleq \{ i \in \llbracket 1, N \rrbracket : H(\overline{U}_i|\overline{U}^{i-1}) \geq \delta_N \}, \\
\mathcal{V}_{\overline{X}} &\triangleq \{ i \in \llbracket 1, N \rrbracket : H(\overline{U}_i|\overline{U}^{i-1}) \geq 1- \delta_N \},\\
\mathcal{H}_{{X}} &\triangleq \{ i \in \llbracket 1, N \rrbracket : H({U}_i|{U}^{i-1}) \geq \delta_N \}, \\
\mathcal{V}_{{X}} &\triangleq \{ i \in \llbracket 1, N \rrbracket : H({U}_i|{U}^{i-1}) \geq 1- \delta_N \}.
\end{align*}
We compress $X^N$ as 
 \begin{align}
 	A \triangleq \left( U^N[\mathcal{V}_{\underline{X}}]  \lVert U^N[\mathcal{H}_{\overline{X}} \backslash \mathcal{V}_{\underline{X}}] \oplus K \right), \label{eqenc}
 \end{align}
where $K$ is a sequence of $|\mathcal{\mathcal{H}_{\overline{X}} \backslash \mathcal{V}_{\underline{X}}}|$ uniformly distributed bits shared between the encoder and decoder. 

\textbf{Step 3.}
 We now repeat $L$ times Step 2 and perform source resolvability with lossless reconstruction of the input from the output.  We choose $M \triangleq N^2$ and let $Y^M$ be a sequence of $M$ independent and identically distributed random variables with distribution $p_Y$. We define $V^M \triangleq G_{2n} Y^M$, for $\beta < 1/2$, $\delta_M \triangleq 2^{-M^{\beta}}$, the set
$$
\mathcal{V}_Y \triangleq \left\{ i \in \llbracket 1, M \rrbracket : H(V_i|V^{i-1}) \geq \delta_M \right\},
$$
and 
\begin{align}
	L \triangleq \left\lfloor \frac{|\mathcal{V}_Y| - |I_{N}|}{|A| } \right\rfloor. \label{eqL}
\end{align}
 We apply Step 2 to $L$ independent sequences $X^N$ to form $A_i$,  $i \in \llbracket 1, L \rrbracket$. Note that this requires $L$ sequences $(K_i)_{i \in \llbracket 1, L \rrbracket}$ of shared randomness between the encoder and the decoder.
We denote the concatenation of these $L$ compressed sequences by $A_{\mathcal{L}} \triangleq  \stackrel[\!i \in \llbracket 1 , L \rrbracket ]{}{\lVert} \!\!\!\! A_i $. 

 Next, we let $R$ be a sequence of $|\mathcal{V}_Y| - L |A| - |I_N|$ uniformly distributed bits (only known by the encoder) and define $\widetilde{V}^M$ as follows.
We set $$\widetilde{V}^M [ \mathcal{V}_Y ] \triangleq (A_{\mathcal{L}} \lVert R \lVert I_N)$$ and successively draw the remaining components of $\widetilde{V}^M$ in $\mathcal{V}_{Y}^c$, according to 
      \begin{equation}  \label{true_distribcr}
\widetilde{p}_{V_j|V^{j-1}} (v_j|\widetilde{V}^{j-1}) \triangleq 
        {p}_{V_j|V^{j-1}} (v_j|\widetilde{V}^{j-1}) \text{ for }  j \in  \mathcal{V}_{Y}^c.
\end{equation}

Finally, the encoder returns $$\widetilde{Y}^M \triangleq G_{2n} \widetilde{V}^M.$$

\textbf{Decoding}. Upon observing $\widetilde{Y}^M$, the decoder computes $\widetilde{V}^M = G_{2n} \widetilde{Y}^M$ and recovers $I_N$ from the last $\lceil \log(q) \rceil$ bits of $\widetilde{V}^M [\mathcal{V}_Y]$. Next, with $K_0$ and $I_N$, the decoder  can recover $\underline{p}$ and $\overline{p}$ (from~\eqref{eqp1}, \eqref{eqp}, and \eqref{eqp2}), determine $\mathcal{H}_{\overline{X}}$ and $\mathcal{V}_{\underline{X}}$, and recover $A_{\mathcal{L}}$ from the first $L |A|$ bits of  $\widetilde{V}^M [\mathcal{V}_Y]$. With $(K_i)_{i \in \llbracket 1,L \rrbracket}$ and $A_{\mathcal{L}}$, the decoder can also recover $(U_i^N[\mathcal{V}_{\underline{X}} \cup \mathcal{H}_{\overline{X}} \backslash \mathcal{V}_{\underline{X}} ])_{i \in \llbracket 1,L \rrbracket}$ by \eqref{eqenc}. Finally, the decoder runs the successive cancellation decoder of \cite{Arikan10} to estimate  $X_i^N$, $i \in \llbracket 1,L \rrbracket$, from $U_i^N[ \mathcal{H}_{\overline{X}} ]$.

\begin{rem}
\eqref{true_distribcr} can be slightly simplified. Specifically, the randomizations could be
replaced by deterministic decisions for $j \in  \mathcal{H}^c_Y$, i.e., randomized decisions are only needed for $j \in  \mathcal{V}^c_Y \backslash \mathcal{H}^c_Y$, as shown
in \cite{chou2015using}.
\end{rem}

\begin{rem}
In the special case of source resolvability, i.e., when the source is known to have a uniform distribution, then no seed is required in our coding scheme.	This was already known, e.g., \cite[Remark 16]{chou2016empirical}. 
\end{rem}
\begin{rem}
In the special case of uniform compression when the distribution of the source to compress is known, polar coding schemes can also be obtained in the presence of side information~\cite{chou2020explicit}.  
\end{rem}

\subsection{Analysis}

\subsubsection{Reliability}
 Note that the estimator $\hat{p}$ used for the parameter $p$ is unbiased and has variance $\sigma^2 = O((LN)^{-1})$. Define the events 
$
\mathcal{E} \triangleq \{ \left(p \geq a_{I_0+1}\right) \text{ or } \left(p \leq a_{I_0-2} \right) \},
$
and 
$
\widetilde{\mathcal{E}} \triangleq \{  (p \leq \hat{p} - N^{-2t} ) \text{ or }  (p \geq \hat{p} + N^{-2t}) \}.
$
 We then have
 \vspace*{-1em}
\begin{align}
\mathbb{P} [ \mathcal{E}] 
 &\stackrel{(a)}{\leq} \mathbb{P} [\widetilde{\mathcal{E}}] \nonumber \\ 
 &= \mathbb{P} [| \hat{p} - p | \geq N^{-2t}]   \nonumber \\
& \stackrel{(b)}{\leq}  \frac{\sigma^2}{N^{-4t^2}} = O(N^{-1+4t^2}L^{-1}), \label{eqlimn}
\end{align}
where $(a)$ holds because for $N$ large enough, $ N^{-2t} = o(a_{I_0-1} - a_{I_0-2}) $, $N^{-2t} = o(a_{I_0+1} - a_{I_0}) $, and $[\hat{p} - N^{-2t}, \hat{p} + N^{-2t}]$ is thus a subinterval of $[a_{I_0-2},a_{I_0+1}]$, and $(b)$ holds by Chebyshev's inequality.
Hence, to show reliability, it is sufficient to consider that $H_b(\overline{p}) \geq H_b(p)$, since $\mathbb{P}[\{H_b(\overline{p}) < H_b(p)\}] = \mathbb{P}[ \{ |1/2-\overline{p}| > |1/2 -p| \}] \leq \mathbb{P} [ \mathcal{E}] \xrightarrow{N \to \infty} 0$ by~\eqref{eqlimn}.%

Recall that when $p$ is known to the encoder and decoder, \cite{Arikan10} shows that it is possible to reconstruct $X^N$ from $U^N[\mathcal{H}_X]$ with error probability bounded by $O(N\delta_N)$, where ${U}^N \triangleq {X}^N G_n$. The following lemma shows that when $p$ is unknown but $H_b(\overline{p}) \geq H_b(p)$, there is no loss of information by compressing $X^N$ as $U^N[\mathcal{H}_{\overline{X}}]$. Moreover, using the successive cancellation decoder of \cite{Arikan10}, by \cite[Lemma 4]{abbe2010universal}, one can reconstruct $X^N$ from $U^N[\mathcal{H}_{\overline{X}}]$ with error probability bounded by $O(N \delta_N)$.
\begin{lem}\label{lemset}
We have $\mathcal{H}_X \subset \mathcal{H}_{\overline{X}}$.%
\end{lem}
\begin{proof}
We closely follow \cite{abbe2010universal}. There exists $\alpha \in [0,1]$ such that $\overline{p} = \alpha (1-p) + (1-\alpha)p$ since we have $H_b(\overline{p})\geq H_b(p)$. Let $B^N$ be a sequence of $N$ identically and independently distributed Bernoulli random variables with parameter $\alpha$ independent of all other random variables. Then, $\widetilde{\overline{X}}^N \triangleq X^N \oplus B^N$ has the same distribution as $\overline{X}^N$. Then, define $\widetilde{\overline{U}}^N \triangleq G_n\widetilde{\overline{X}}^N $. We have for any $ i \in \llbracket 1, N \rrbracket$, $H(\widetilde{\overline{U}}_i|\widetilde{\overline{U}}^{i-1}) \geq H(\widetilde{\overline{U}}_i|\widetilde{\overline{U}}^{i-1} B^N) = H({{U}}_i|{{U}}^{i-1} B^N) = H(U_i|U^{i-1}) $. We conclude that
$\mathcal{H}_X \subset \mathcal{H}_{\overline{X}}$. %
\end{proof}

Consequently, the decoding scheme of Section \ref{sec:decp} succeeds in reconstructing $X^{NL}$ with error probability bounded by $O(NL \delta_N)$, which vanishes as $N \to \infty $ since $L = O(N)$ by \eqref{eqL}.

\subsubsection{Covertness}

Similar to the analysis of reliability, to show that the covertness condition holds in probability, it is sufficient to show covertness when $H_b(\underline{p}) \leq H_b(p)$. In this case, similar to Lemma~\ref{lemset}, we have the following lemma.
\begin{lem} \label{lemset2}
	We have $ \mathcal{V}_{\underline{X}} \subset \mathcal{V}_X $ and $\mathcal{H}_{\underline{X}} \subset  \mathcal{H}_{\overline{X}}$.
\end{lem}
Hence, by Lemma \ref{lemset2}, $\mathcal{V}_{\underline{X}} \subset \mathcal{H}_{\underline{X}} \subset  \mathcal{H}_{\overline{X}}$, and $|A| =|\mathcal{H}_{\overline{X}}| $. Then, let $p_{U_{\mathcal{V}_{\underline{X}}}}$ and $p_{U_{\mathcal{H}_{\overline{X}}}}$ be the uniform distribution over $\llbracket 1,2^{|\mathcal{V}_{\underline{X}}|} \rrbracket$ and $\llbracket 1,2^{|\mathcal{H}_{\overline{X}}|} \rrbracket$, respectively. Observe that $A_i$, $i \in \llbracket 1,L\rrbracket$, is nearly uniform in the sense that
\begin{align} \label{eqsingle}
{D} (p_{A_i} \lVert p_{U_{\mathcal{H}_{\overline{X}}}})
& \stackrel{(a)}=  {D} (U^N[\mathcal{V}_{\underline{X}}] \lVert p_{U_{\mathcal{V}_{\underline{X}}}}) \nonumber \\
& =  |\mathcal{V}_{\underline{X}}| - H(U^N[\mathcal{V}_{\underline{X}}] ) \nonumber \displaybreak[0] \\
& \stackrel{(b)}\leq  |\mathcal{V}_{\underline{X}}| - \sum_{j \in \mathcal{V}_{\underline{X}}} H(U_j | U^{j-1}) \nonumber  \displaybreak[0] \\
& \stackrel{(c)}\leq  |\mathcal{V}_{\underline{X}}| - \sum_{j \in \mathcal{V}_{\underline{X}}} (1-\delta_N) \nonumber \\
& \leq N \delta_N,
\end{align}
where $(a)$ holds by the chain rule for KL-divergence and uniformity of $K$, $(b)$ holds by the chain rule and because conditioning reduces entropy, $(c)$ holds by definition of  $ \mathcal{V}_{\underline{X}} $.

Next, define  $p_{U_{\mathcal{L}}}$ the uniform distribution over $\llbracket 1,2^{L|\mathcal{H}_{\overline{X}}|} \rrbracket$ and define $\bar{V}^M$ similar to $\widetilde{V}^M$ but by replacing $A_{\mathcal{L}}$ in the description of Step 3 in Section \ref{sec:decp} by a sequence distributed according to $p_{U_{\mathcal{L}}}$. We have  
\begin{align}
\mathbb{D} (\widetilde{p}_{V^M } \lVert  \bar{p}_{V^M })
& \stackrel{(a)}{\leq} 
\mathbb{D} (\widetilde{p}_{V^M A_{\mathcal{L}}} \lVert  \bar{p}_{V^M U_{\mathcal{L}}}) \nonumber \\
& = \mathbb{D} (\widetilde{p}_{V^M | A_{\mathcal{L}}} p_{A_{\mathcal{L}}} \lVert  \bar{p}_{V^M | U_{\mathcal{L}}} p_{ U_{\mathcal{L}}}) \nonumber \\
& \stackrel{(b)}{\leq}  \mathbb{D} ( p_{A_{\mathcal{L}}} \lVert  p_{U_{\mathcal{L}}}) \nonumber \\
& \stackrel{(c)}{=} \sum_{i=1}^L  \mathbb{D} (p_{A_{i}} \lVert  p_{U})  \nonumber \\
&\stackrel{(d)}{\leq} LN \delta_N, \label{eqmult} 
\end{align}
where $(a)$ holds by the chain rule and positivity of the KL-divergence, $(b)$ holds by the chain rule and since $\bar{V}^M$ and $\widetilde{V}^M$ are produced similarly given $U_{\mathcal{L}}$ or $A_{\mathcal{L}}$, $(c)$ holds by the chain rule, $(d)$ holds by \eqref{eqsingle}.

Finally, we have
 \begin{align*}
&\mathbb{D} (\widetilde{p}_{Y^M } \lVert  {p}_{Y^M})\\
&=\mathbb{D} (\widetilde{p}_{V^M } \lVert  {p}_{V^M})\\
& \stackrel{(a)}{\leq}  M \log \mu_V^{-1} \sqrt{2 \ln 2} \left[ \sqrt{\mathbb{D} (\widetilde{p}_{V^M } \lVert  \bar{p}_{V^M})} + \sqrt{\mathbb{D} (\bar{p}_{V^M} \lVert  {p}_{V^M})} \right]\\
&  \stackrel{(b)}{\leq}   M \log \mu_V^{-1} \sqrt{2 \ln 2} \left[ \sqrt{ LN \delta_N} + \sqrt{ M\delta_M} \right]\\
& \xrightarrow{N\to \infty} 0,
\end{align*}
where $(a)$ holds by Lemma \ref{corent} with $\mu_V \triangleq \min_{v \in \mathcal{V}}p_V(v)$, $(b)$ holds by \eqref{eqmult} and because $ \mathbb{D} (\bar{p}_{V^M} \lVert  {p}_{V^M}) \leq  M\delta_M$, which can be shown by similar arguments to \cite[Lemma 1]{chou2016empirical}, and the limit holds since $L = O(N)$ and $M= N^2$.

\subsubsection{Compression rate}
By definition of $L$, there exists $r \in \llbracket 0 , |\mathcal{H}_{\overline{X}}| -1 \rrbracket$ such that
$L |\mathcal{H}_{\overline{X}}| + r = |\mathcal{V}_Y| - |I_N|.$ We deduce
\begin{align*}
\frac{LN}{M}
 &= \frac{|\mathcal{V}_Y|/M}{|\mathcal{H}_{\overline{X}}|/N} - \frac{|I_N|/M + r/M}{|\mathcal{H}_{\overline{X}}|/N}  \\
 & \leq \frac{|\mathcal{V}_Y|/M}{|\mathcal{H}_{\overline{X}}|/N} \\
 & \xrightarrow{N \to \infty} \frac{H(Y)}{H(X)},
\end{align*}
where the limit holds in probability because $\lim_{N \to \infty}|\mathcal{V}_Y|/M=H(Y)$ by \cite[Lemma~7]{chou2014polar}, $\lim_{N \to \infty}|\mathcal{H}_X|/N=H(X)$ by \cite{Arikan10}, and $\plim_{N \to \infty} H_b(\bar{p}) = H_b(p)$ by \eqref{eqlimn}. %
\subsubsection{Length of the shared seed}
Finally, we verify that the length of the shared seed that is needed in the coding scheme of Section \ref{sec:decp} is negligible compared to the total length $LN$ of the sequence that is compressed.
Note that in Step 2 $|K|=o(N)$ since $|\mathcal{\mathcal{H}_{\overline{X}} \backslash \mathcal{V}_{\underline{X}}}| = |\mathcal{\mathcal{H}_{\overline{X}} | -| \mathcal{V}_{\underline{X}}}|$ (because $\mathcal{V}_{\underline{X}} \subset \mathcal{H}_{\overline{X}}$ by Lemma~\ref{lemset2}) and $\lim_{N\to \infty} |\mathcal{H}_{\overline{X}} |/N = \lim_{N\to \infty} H_b(\overline{p})  = \lim_{N\to \infty} H_b(\underline{p})= \lim_{N\to \infty} |\mathcal{V}_{\underline{X}} |/N$ (by~\cite{Arikan10} and \cite[Lemma~1]{Chou14rev}). Hence, the total length of the shared seed is $\sum_{i=0}^L |K_i| = |K_0| + L|K| = o(LN).$

\section{Concluding Remarks} \label{sec:concl}
We have introduced the notion of universal covertness for DMSs to generalize information-theoretic steganography \cite{cachin2004information}, uniform lossless source coding~\cite{Chou13}, source resolvability \cite{HanBook}, and random number conversion~\cite{HanBook,kumagai2013new}. %

Our proposed coding scheme consists of the combination of (i) a type-based coding scheme able to simultaneously perform universal lossless source coding and ensure an almost uniform encoder output, and (ii) source resolvability with lossless reconstruction of the input from the output.
Our coding scheme uses a seed, i.e., a uniformly distributed sequence of bits, shared by the encoder and the decoder.
Although our seed rate vanishes as $n$ grows to infinity and has a length $\Theta ( \log n)$, it is not clear  whether a smaller seed could offer similar convergence rates. %

Finally, we have proposed an explicit low-complexity encoding and decoding  scheme for~universal covertness of binary memoryless sources based on polar codes. 
Our coding scheme requires a seed length that grows faster than $\log n$, yet, its rate still vanishes as $n$~grows. Note that in the special case of source resolvability, i.e., when the source is known to have a uniform distribution, then no seed is required in our coding scheme.

\section*{Acknowledgment}
The authors would like to thank the Associate Editor and the anonymous reviewers for their valuable comments. In particular, the authors thank one reviewer for suggesting an alternative achievability scheme that led to Theorem \ref{Cor2} and established sufficiency of a logarithmic amount of common randomness.

\appendices

\section{Proof of Lemma \ref{lem_V}} \label{App_V}
Fix $p_{\bar{X}} \in \mathcal{P}_n(\mathcal{X})$. We write $b_n$, $V_n$, $\bar{V}_n$ instead of $b_n(p_{\bar{X}})$, $V_n(p_{\bar{X}})$, $\bar{V}_n(p_{\bar{X}})$, respectively, to simplify the notation. By Euclidean division, there exist $q \in \mathbb{N}$, $r \in  \llbracket 0, b_n -1\rrbracket$ such that $a_n-1 = b_n q +r$. Next, for $v \in \llbracket 0, b_n -1\rrbracket$, we have
\begin{align*}
	p_{V_n}(v) 
	& =  \mathbb{P}[U_1\mod b_n = v] \\
    & =  \sum_{u = 0}^{a_n -1} p_{U_1}(u)\mathds{1} \{ u \mod b_n = v\}\\
	&= \frac{1}{a_n} \sum_{u = 0}^{a_n -1}\mathds{1} \{ u \mod b_n = v\}\\
	& = \begin{cases} \frac{q+1}{a_n} &  \text{ if }v \leq r\\ \frac{q}{a_n} &  \text{ if }v > r  \end{cases}.
\end{align*}
Then, we have
\begin{align*}
&\mathbb{V}\left(p_{V_n},p_{\bar{V}_n}\right) \\
& = \sum_{v \in \llbracket 0, b_n -1\rrbracket}	 |p_{V_n}(v)-p_{\bar{V}_n}(v)|\\
& = \sum_{v \in \llbracket 0, r\rrbracket}	 \left\lvert \frac{q+1}{a_n}-\frac{1}{b_n}\right\rvert + \sum_{v \in \llbracket r+1, b_n-1\rrbracket}	 \left\lvert\frac{q}{a_n}-\frac{1}{b_n}\right\rvert\\
& = \sum_{v \in \llbracket 0, r\rrbracket}	 \left\lvert \frac{(q+1)b_n-a_n}{a_n b_n}\right\rvert + \sum_{v \in \llbracket r+1, b_n-1\rrbracket}	 \left\lvert\frac{q b_n-a_n}{a_n b_n}\right\rvert\\
& = \sum_{v \in \llbracket 0, r\rrbracket}	 \left\lvert \frac{b_n - r - 1}{a_n b_n}\right\rvert + \sum_{v \in \llbracket r+1, b_n-1\rrbracket}	 \left\lvert\frac{r+1}{a_n b_n}\right\rvert\\
& = (r+1)	 O \left( \frac{1}{a_n}\right) + (b_n-1 -r ) O \left( \frac{1}{a_n}\right)\\
& = b_n	 O \left( \frac{1}{a_n}\right) \\
&  =  O\left(\frac{1}{n^{\alpha-\beta}}\right).
\end{align*}

\section{Proof of Lemma \ref{lemcond}} \label{App_lemcond}

 Fix $p_{\bar{X}} \in \mathcal{P}_n(\mathcal{X})$. We write $b_n$, $c_n$, $V_n$, $\bar{V}_n$, $p_{U}$, $k$, instead of $b_n(p_{\bar{X}})$, $c_n(p_{\bar{X}})$,$V_n(p_{\bar{X}})$, $\bar{V}_n(p_{\bar{X}})$, $p_{U(p_{\bar{X}})}$, $k_{s(p_{\bar{X}})}$, respectively, to simplify the notation.
	 
	 Define $\bar{\phi}^{(1)}_n(X^n,U_{d_n})$ as $\phi^{(1)}_n(X^n,U_{d_n})$ when $V_n$ is replaced by $\bar{V}_n$. We have
\begin{align}
	&\mathbb{V} \left( p_{\bar{\phi}^{(1)}_n(X^n,U_{d_n})|X^n \in T_{\bar{X}}^n } , p_{U}\right) \nonumber \\ \nonumber
	& = \sum_{m \in \mathcal{M}} \left\lvert p_{\bar{\phi}^{(1)}_n(X^n,U_{d_n})| X^n \in T_{\bar{X}}^n}(m) - p_{U}(m)\right\rvert \\ \nonumber
	& = {\sum_{m \in \mathcal{M}} \left\lvert  \frac{1}{|T_{\bar{X}}^n|b_n}  \sum_{v=0}^{b_n-1} \sum_{x^n \in T_{\bar{X}}^n}\mathds{1} \{ k(x^n) + v |T_{\bar{X}}^n|=m \} \right. }\\\nonumber
	&\smash{\left. \phantom{---------------\frac{1}{|T_{\bar{X}}^n|b_n}  \sum_{v=0}^{b_n-1} }  - p_{U}(m) \right\rvert} \\ \nonumber
	& \stackrel{(a)}\leq \sum_{m = c_n -|T_{\bar{X}}^n| + 1}^{c_n-1} \max \left( p_{U}(m), \left\lvert\frac{1}{|T_{\bar{X}}^n|b_n}  - p_{U}(m) \right\rvert\right) \\\nonumber
	& \phantom{--} + \sum_{m =0}^{c_n -|T_{\bar{X}}^n|} \left\lvert  \frac{1}{|T_{\bar{X}}^n|b_n}   - p_{U}(m) \right\rvert \\ \nonumber
	& \stackrel{(b)}\leq \sum_{m = c_n  -|T_{\bar{X}}^n| + 1}^{c_n-1}  p_{U}(m)  + \sum_{m =0}^{c_n -1} \left\lvert  \frac{1}{|T_{\bar{X}}^n|b_n}   - p_{U}(m) \right\rvert \\ \nonumber
	& \leq \frac{|T_{\bar{X}}^n|}{c_n}    +  \left\lvert  \frac{c_n}{|T_{\bar{X}}^n|b_n}   - 1 \right\rvert \\ \nonumber
	& \stackrel{(c)} \leq \frac{1}{n^{\beta}} + \frac{c_n}{c_n-|T_{\bar{X}}^n|} -1 \\ \nonumber
	& = \frac{1}{n^{\beta}} + \frac{1}{ \frac{c_n}{|T_{\bar{X}}^n|}-1}\\ \nonumber
	& \stackrel{(d)} \leq  \frac{1}{n^{\beta}} + \frac{1}{ n^{\beta}-1}\\
	& = O\left(\frac{1}{n^{\beta}}\right), \label{eq1}
\end{align}
where $(a)$ holds because  for $v \in \llbracket 0 , b_n -1 \rrbracket$, $x^n \in T_{\bar{X}}^n$,  $k(x^n) + v |T_{\bar{X}}^n| \in \llbracket 0 , b_n|T_{\bar{X}}^n| -1  \rrbracket$ and $
 b_n |T_{\bar{X}}^n| \in \llbracket c_n -|T_{\bar{X}}^n|  + 1, c_n \rrbracket
$, $(b)$ holds because for any $x,y\geq 0$, $\max(x,y) \leq x+y$, $(c)$ and $(d)$ holds because $c_n \in \llbracket n^{\beta} |T_{\bar{X}}^n|, 2 n^{\beta} |T_{\bar{X}}^n| -1\rrbracket$ and $b_n |T_{\bar{X}}^n| \in \llbracket c_n -|T_{\bar{X}}^n|  + 1, c_n\rrbracket$.
Next, we have
\begin{align*}
	&\mathbb{V} \left(p_{{\phi}^{(1)}_n(X^n,U_{d_n})| X^n \in T_{\bar{X}}^n}, p_{U}\right) \\
	& \stackrel{(a)}\leq \mathbb{V} \left(p_{{\phi}^{(1)}_n(X^n,U_{d_n})| X^n \in T_{\bar{X}}^n},p_{\bar{\phi}^{(1)}_n(X^n,U_{d_n})| X^n \in T_{\bar{X}}^n} \right) \\
	& \phantom{--}+  \mathbb{V} \left(p_{\bar{\phi}^{(1)}_n(X^n,U_{d_n})| X^n \in T_{\bar{X}}^n} , p_{U}\right)\\
	& \stackrel{(b)} \leq \mathbb{V} \left(p_{V_n},p_{\bar{V}_n} \right) +  \mathbb{V} \left(p_{\bar{\phi}^{(1)}_n(X^n,U_{d_n})| X^n \in T_{\bar{X}}^n},  p_{U}\right) \\
	& \stackrel{(c)} \leq O\left(\frac{1}{n^{\min(\alpha - \beta, \beta)}}\right),
\end{align*}
where $(a)$ holds by the triangle inequality, $(b)$ holds by the data processing inequality for the variational distance, $(c)$ holds by~\eqref{eq1} and Lemma \ref{lem_V}.

\section{Proof of Lemma \ref{lem_avg}} \label{App_avg}
We have 	
\begin{align*}
&\mathbb{V} \left( p_{\phi^{(1)}_n(X^n,U_{d_n})} , p_{U(X^n)}\right)		\\
&= \sum_{m=0}^{c_n(X^N)} \left\lvert p_{\phi^{(1)}_n(X^n,U_{d_n})}(m) - p_{U(X^n)}(m)\right\lvert \\
&= \sum_{m=0}^{c_n(X^N)} \left\lvert  \sum_{p_{\bar{X}} \in \mathcal{P}_n(\mathcal{X})} \mathbb{P} [X^n \in T_{\bar{X}}^n] \right. \\
& \left. \phantom{\sum_{m=0}^{c_n(X^N)}--}  \times \left( p_{\phi^{(1)}_n(X^n,U_{d_n})| X^n \in T_{\bar{X}}^n}(m)  - p_{U(p_{\bar{X}})} (m)\right)\right\lvert \\
& \stackrel{(a)} \leq  \sum_{p_{\bar{X}}\in \mathcal{P}_n(\mathcal{X})} \mathbb{P} [X^n \in T_{\bar{X}}^n]   \\
&  \phantom{\sum_{m=0}^{c_n(X^N)}--}  \times  \sum_{m=0}^{c_n(X^N)}\left\lvert    p_{\phi^{(1)}_n(X^n,U_{d_n})| X^n \in T_{\bar{X}}^n}(m) - p_{U(p_{\bar{X}})} (m)\right\lvert \\
& \stackrel{(b)}  \leq O\left(\frac{1}{n^{\min(\alpha - \beta, \beta)}}\right),
\end{align*}
where $(a)$ holds by the triangle inequality, $(b)$ holds by Lemma~\ref{lemcond}.

\section{Proof of Lemma \ref{lem_div}} \label{App_div}
	We have 
	\begin{align*}
&\mathbb{D} \left( p_{\phi^{(1)}_n(X^n,U_{d_n})} \lVert  p_{U(X^n)}\right)		\\
&=  \log c_n(X^n) - H(\phi^{(1)}_n(X^n,U_{d_n})) \\
& \stackrel{(a)}\leq - \mathbb{V}\!\! \left( p_{\phi^{(1)}_n(X^n,U_{d_n})} , p_{U(X^n)} \!\! \right) \log \mathbb{V}\!\! \left( p_{\phi^{(1)}_n(X^n,U_{d_n})} , p_{U(X^n)}\!\!\right) \\
& \phantom{--}+ \mathbb{V} \left( p_{\phi^{(1)}_n(X^n,U_{d_n})} , p_{U(X^n)}\right) \log c_n(X^n) \\ 
& \stackrel{(b)}\leq  O\left(\frac{\log n}{n^{\min(\alpha - \beta, \beta)}}\right) +  O\left(\frac{n + \log n}{n^{\min(\alpha - \beta, \beta)}}\right)\\
& =  O\left(\frac{1}{n^{\min(\alpha - \beta, \beta)-1}}\right),
\end{align*}
where $(a)$ holds by \cite[Lemma 2.7]{bookCsizar}, $(b)$ holds by Lemma~\ref{lem_avg}.

\section{Proof of asymptotic rate optimality in Theorem \ref{Cor9} } \label{App_converse}

Consider an $(n,2^{d_n})$ variable-length universal covert source code for a DMS $(\mathcal{X},p_X)$ with respect to the DMS $(\mathcal{Y},p_Y)$ and let $M$ denote the encoder output. Assume that $\displaystyle\lim_{n \to \infty} \mathbb{D} \left(p_{M} \lVert  p^{\otimes m(X^n)}_{Y} \right) = 0$, and $\lim_{n \to \infty} \mathbb{P}[\widehat{X}^n \neq X^n]=0$, where $\widehat{X}^n$ denotes the estimate of $X^n$ formed by the decoder. We have
 \begin{align*}
& m(X^n)H(Y) \\
& = H(Y^{m(X^n)}) \\ 
&  \stackrel{(a)}{=} H(M) - D \displaybreak[0] \\
&  \stackrel{(b)}{=} H(M X^nU_{d_n}) - H(X^n|MU_{d_n}) - H(U_{d_n}|M) - D   \displaybreak[0] \\
&  \stackrel{(c)}{\geq}  H(X^n) + H(U_{d_n}) - H(X^n|MU_{d_n}) - H(U_{d_n}|M) - D  \displaybreak[0] \\
&  \stackrel{(d)}{\geq}  H(X^n) - H(X^n|MU_{d_n})  - D  \\
& \stackrel{(e)}{\geq}  nH(X) - (1 + n \epsilon_n)  - D ,
\end{align*}
where in $(a)$ we have defined $D \triangleq H(M)-H(Y^{m(X^n)})$, $(b)$ holds by the chain rule, $(c)$ holds by independence between  $X^n$ and $U_{d_n}$, $(d)$ holds because $I(U_{d_n};M)\geq 0$, $(e)$ holds by Fano's inequality with $\epsilon_n$ such that $\lim_{n \to \infty} \epsilon_n =0$ because $H(X^n|MU_{d_n})\leq H(X^n|\widehat{X}^n)$.

Hence, we obtain
 \begin{align} \label{eqff}
 \frac{{m(X^n)}}{n} \geq \frac{H(X)}{H(Y)} - 
 \frac{1}{nH(Y)} - \frac{\epsilon_n}{H(Y)}  - \frac{D}{nH(Y)}.
\end{align}
Next, define the function $$g:(x,y)\mapsto  \sqrt{2 \ln2} \sqrt{x} \log \left(y/(\sqrt{2 \ln2} \sqrt{x})\right).$$ Using~\cite[Lemma 2.7]{bookCsizar}, Pinsker's inequality, and the fact that $x \mapsto -x \log x$ is increasing over $]0,e^{-1}]$, we have
 \begin{align}\label{eqfff}
\displaystyle\plim_{n \to \infty} \frac{D}{nH(Y)} 
\leq \displaystyle\plim_{n \to \infty} \frac{m(X^n)}{n}\frac{g\left(\mathbb{D} \left(p_{M} \lVert  p^{\otimes m(X^n)}_{Y} \right),|\mathcal{Y}| \right)}{H(Y)}.
\end{align}
Finally, by combining \eqref{eqff}, \eqref{eqfff}, and the hypothesis $\displaystyle\lim_{n \to \infty} \mathbb{D} \left(p_{M} \lVert  p^{\otimes m(X^n)}_{Y} \right)=0 $, we obtain
$
 \displaystyle\plim_{n \to \infty}  \frac{m(X^n)}{n} \geq \frac{H(X)}{H(Y)} .
$

 \bibliographystyle{IEEEtran}
\bibliography{polarwiretap}

\begin{thebibliography}{10}
\providecommand{\url}[1]{#1}
\csname url@samestyle\endcsname
\providecommand{\newblock}{\relax}
\providecommand{\bibinfo}[2]{#2}
\providecommand{\BIBentrySTDinterwordspacing}{\spaceskip=0pt\relax}
\providecommand{\BIBentryALTinterwordstretchfactor}{4}
\providecommand{\BIBentryALTinterwordspacing}{\spaceskip=\fontdimen2\font plus
\BIBentryALTinterwordstretchfactor\fontdimen3\font minus
  \fontdimen4\font\relax}
\providecommand{\BIBforeignlanguage}[2]{{%
\expandafter\ifx\csname l@#1\endcsname\relax
\typeout{** WARNING: IEEEtran.bst: No hyphenation pattern has been}%
\typeout{** loaded for the language `#1'. Using the pattern for}%
\typeout{** the default language instead.}%
\else
\language=\csname l@#1\endcsname
\fi
#2}}
\providecommand{\BIBdecl}{\relax}
\BIBdecl

\bibitem{chou2016universal}
R.~Chou, M.~Bloch, and A.~Yener, ``Universal covertness for discrete memoryless
  sources,'' in \emph{Annual Allerton Conf. on Communication Control and
  Computing}, Monticello, IL, 2016, pp. 516--523.

\bibitem{maurer2000authentication}
U.~Maurer, ``Authentication theory and hypothesis testing,'' \emph{IEEE Trans.
  Inf. Theory}, vol.~46, no.~4, pp. 1350--1356, 2000.

\bibitem{blahut1987principles}
R.~Blahut, \emph{Principles and practice of information theory}.\hskip 1em plus
  0.5em minus 0.4em\relax Addison-Wesley Longman Publishing Co., Inc., 1987.

\bibitem{Chou13}
R.~Chou and M.~Bloch, ``Data compression with nearly uniform output,'' in
  \emph{Proc. of IEEE Int. Symp. Inf. Theory}, Istanbul, Turkey, 2013, pp.
  1979--1983.

\bibitem{HanBook}
T.~Han, \emph{Information-{S}pectrum {M}ethods in {I}nformation
  {T}heory}.\hskip 1em plus 0.5em minus 0.4em\relax Springer, 2002, vol.~50.

\bibitem{kumagai2013new}
W.~Kumagai and M.~Hayashi, ``Second-order asymptotics of conversions of
  distributions and entangled states based on {R}ayleigh-normal probability
  distributions,'' \emph{{IEEE} {T}rans. {I}nf. {T}heory}, vol.~63, no.~3, pp.
  1829--1857, 2017.

\bibitem{bookCsizar}
I.~Csisz\'{a}r and J.~K\"{o}rner, \emph{Information Theory: Coding Theorems for
  Discrete Memoryless Systems}.\hskip 1em plus 0.5em minus 0.4em\relax
  Cambridge Univ Pr, 1981.

\bibitem{oohama07}
Y.~Oohama, ``Intrinsic randomness problem in the framework of {S}lepian-{W}olf
  separate coding system,'' \emph{{IEICE} Trans. On Fundamentals Of Electronics
  Communications And Computer Sciences {E} Series {A}}, vol.~90, no.~7, p.
  1406, 2007.

\bibitem{bash2013limits}
B.~Bash, D.~Goeckel, and D.~Towsley, ``Limits of reliable communication with
  low probability of detection on awgn channels,'' \emph{IEEE J. Sel. Areas
  Commun.}, vol.~31, no.~9, pp. 1921--1930, 2013.

\bibitem{bloch2015covert}
M.~Bloch, ``Covert communication over noisy channels: A resolvability
  perspective,'' \emph{IEEE Trans. Inf. Theory}, vol.~62, no.~5, pp.
  2334--2354, 2016.

\bibitem{wang2015fundamental}
L.~Wang, G.~Wornell, and L.~Zheng, ``Fundamental limits of communication with
  low probability of detection,'' \emph{IEEE Trans. Inf. Theory}, vol.~62,
  no.~6, pp. 3493--3503, 2016.

\bibitem{wang2008perfectly}
Y.~Wang and P.~Moulin, ``Perfectly secure steganography: Capacity, error
  exponents, and code constructions,'' \emph{IEEE Trans. Inf. Theory}, vol.~54,
  no.~6, pp. 2706--2722, 2008.

\bibitem{cachin2004information}
C.~Cachin, ``An information-theoretic model for steganography,''
  \emph{Information and Computation}, vol. 192, no.~1, pp. 41--56, 2004.

\bibitem{Han05}
T.~Han, ``Folklore in source {c}oding: {I}nformation-{s}pectrum {a}pproach,''
  \emph{{IEEE} {T}rans. {I}nf. {T}heory}, vol.~51, no.~2, pp. 747--753, 2005.

\bibitem{Hayashi08}
M.~Hayashi, ``Second-order asymptotics in fixed-length source coding and
  intrinsic randomness,'' \emph{{IEEE} {T}rans. {I}nf. {T}heory}, vol.~54,
  no.~10, pp. 4619--4637, 2008.

\bibitem{Wyner75}
A.~D. Wyner, ``The wire-tap channel,'' \emph{{T}he {B}ell {S}ystem {T}echnical
  {J}ournal}, vol.~54, no.~8, pp. 1355--1387, 1975.

\bibitem{Csiszar78}
I.~Csisz{\'a}r and J.~Korner, ``Broadcast {C}hannels with {C}onfidential
  {M}essages,'' \emph{{IEEE} {T}rans. {I}nf. {T}heory}, vol.~24, no.~3, pp.
  339--348, 1978.

\bibitem{sason2015bounds}
I.~Sason and S.~Verd{\'u}, ``$ f $-divergence inequalities,'' \emph{IEEE Trans.
  Inf. Theory}, vol.~62, no.~11, pp. 5973--6006, 2016.

\bibitem{chou2016empirical}
R.~Chou, M.~Bloch, and J.~Kliewer, ``Empirical and strong coordination via soft
  covering with polar codes,'' \emph{IEEE Trans. Inf. Theory}, vol.~64, no.~7,
  pp. 5087--5100, 2018.

\bibitem{vembu1995generating}
S.~Vembu and S.~Verd{\'u}, ``Generating random bits from an arbitrary source:
  {F}undamental limits,'' \emph{{IEEE} {T}rans. {I}nf. {T}heory}, vol.~41,
  no.~5, pp. 1322--1332, 1995.

\bibitem{yagi2017variable}
H.~Yagi and T.~S. Han, ``Variable-length resolvability for general sources,''
  in \emph{{IEEE} {I}nt. {S}ymp. {I}nf. {T}heory}, Aachen, Germany, 2017, pp.
  1748--1752.

\bibitem{Chou14b}
R.~Chou and M.~Bloch, ``Uniform distributed source coding for the multiple
  access wiretap channel,'' in \emph{Proc. IEEE Conf. on Communications and
  Network Security (CNS)}, San Francisco, CA, 2014, pp. 127--132.

\bibitem{Chou15d}
B.~Vellambi, M.~Bloch, R.~Chou, and J.~Kliewer, ``Lossless and {L}ossy {S}ource
  {C}ompression with {N}ear-{U}niform {O}utputs: {I}s {C}ommon {R}andomness
  {A}lways {R}equired?'' in \emph{Proc. IEEE Int. Symp. Inf. Theory}, 2015.

\bibitem{oohama1994universal}
Y.~Oohama and T.~Han, ``Universal coding for the {S}lepian-{W}olf data
  compression system and the strong converse theorem,'' \emph{IEEE Trans. Inf.
  Theory}, vol.~40, no.~6, pp. 1908--1919, 1994.

\bibitem{basharin1959statistical}
G.~Basharin, ``On a statistical estimate for the entropy of a sequence of
  independent random variables,'' \emph{Theory of Probability \& Its
  Applications}, vol.~4, no.~3, pp. 333--336, 1959.

\bibitem{yassaee2014achievability}
M.~Yassaee, M.~Aref, and A.~Gohari, ``Achievability proof via output statistics
  of random binning,'' \emph{IEEE Trans. Inf. Theory}, vol.~60, no.~11, pp.
  6760--6786, 2014.

\bibitem{nafea2018new}
M.~Nafea and A.~Yener, ``A new wiretap channel model and its strong secrecy
  capacity,'' \emph{IEEE Trans. Inf. Theory}, vol.~64, no.~3, pp. 2077--2092,
  2018.

\bibitem{Arikan10}
E.~Arikan, ``Source {P}olarization,'' in \emph{Proc. of IEEE Int. Symp. Inf.
  Theory}, Austin, TX, 2010, pp. 899--903.

\bibitem{Chou14prep}
R.~Chou, B.~Vellambi, M.~Bloch, and J.~Kliewer, ``Coding schemes for achieving
  strong secrecy at negligible cost,'' \emph{IEEE Trans. Inf. Theory}, vol.~63,
  no.~3, pp. 1858--1873, 2017.

\bibitem{chou2020explicit}
R.~Chou, ``Unified framework for polynomial-time wiretap channel codes,''
  \emph{arXiv preprint arXiv:2002.01924}, 2020.

\bibitem{chou2014polar}
R.~Chou and M.~Bloch, ``Polar coding for the broadcast channel with
  confidential messages: {A} random binning analogy,'' \emph{{IEEE} {T}rans.
  {I}nf. {T}heory}, vol.~62, no.~5, pp. 2410--2429, 2016.

\bibitem{chou2015using}
R.~Chou and M.~R. Bloch, ``Using deterministic decisions for low-entropy bits
  in the encoding and decoding of polar codes,'' in \emph{Annual Allerton Conf.
  on Communication, Control, and Computing}.\hskip 1em plus 0.5em minus
  0.4em\relax IEEE, 2015, pp. 1380--1385.

\bibitem{abbe2010universal}
E.~Abbe, ``Universal polar coding and sparse recovery,'' \emph{arXiv preprint
  arXiv:1012.0367}, 2010.

\bibitem{Chou14rev}
R.~Chou, M.~Bloch, and E.~Abbe, ``Polar coding for secret-key generation,''
  \emph{{IEEE} {T}rans. {I}nf. {T}heory}, vol.~61, no.~11, pp. 6213--6237,
  2015.

\end{thebibliography}
\end{document}